\documentclass[12pt]{amsart}
\usepackage[utf8]{inputenc}
\usepackage{color}
\usepackage{amsthm}
\usepackage{amssymb}
\usepackage{graphicx}
\usepackage{subfigure,subfig}
\usepackage{dirtytalk}
\usepackage[onehalfspacing]{setspace}
\usepackage[authoryear]{natbib}
\usepackage[ruled,vlined]{algorithm2e}
\usepackage[a4paper, top=1in, left = 1in, right=1in]{geometry}
\makeatletter
\DeclareMathOperator*{\argmax}{arg\,max}

\newtheorem{thm}{Theorem}[section]

\newtheorem{definition}{Definition}[section]

\theoremstyle{plain}

\numberwithin{equation}{section}
\numberwithin{figure}{section}
\theoremstyle{plain}
\AtBeginDocument{
  
}
\makeatother
\usepackage{babel}

\begin{document}

\title{Stable and extremely unequal}
\author{Alfred Galichon, Octavia Ghelfi, Marc Henry}
\address{NYU+SciencesPo, Amazon and Penn State}
\thanks{The first version is dated June 8, 2021. This version is of \today. Ghelfi's contribution reflects work done at New York University, before joining Amazon. The authors thank Laura Doval, Federico Echenique, Larry Samuelson, Ran Shorrer, Olivier Tercieux and an anonymous referee for helpful comments. Galichon acknowledges support from European Research Council grant ERC-CoG No. 866274 and NSF grant DMS-1716489. Ghelfi acknowledges support from NYU’s Henry M. MacCracken Fellowship and NSF grant DMS-1716489. The usual disclaimer applies.}

%\begin{abstract}
%
%Stable allocations are often called ``fair,'' due to the fact that stability eliminates all justifiable envy. In spite of this, we show how stability as a solution concept often comes at the cost of extreme forms on inequality. 
%Restricting our attention to aligned preferences, we show that the stable matching results from the lexicographic welfare maximization of the pairs' welfare, starting with the best-off. We compare this solution with an alternative allocation, that although unstable, maximizes the welfare lexicographically starting with the worst-off pairs. 
%
%\end{abstract}
\maketitle
\doublespacing

\section{Introduction}

In this note, we highlight a simple tension between stability and equality in matching with non transferable utilities. We consider many-to-one matchings and refer to the two sides of the market as students and schools. The latter have aligned preferences, as in \cite{NY:2009}, which in this context means that a school's utility is the sum of its students' utilities. A special case of aligned preferences, known as spatial, arises when utilities are determined by commuting distance to school. 

We show existence and uniqueness of a stable one-to-many matching, under similar assumptions to the ones used by \cite{Eeckhout}, \cite{Clark:2006} and \cite{NY:2009} to prove existence and uniqueness of an equilibrium in the one-to-one case. This stable matching can be obtained through the Deferred Acceptance Algorithm (DAA) of \cite{GS:62}. 

Stable matchings eliminate all justifiable envy, hence are sometimes called ``fair''. For instance~\cite{Kojima_Manea} point out that, due to the fact that stability is regarded as a normative fairness criterion, it is used in many practical assignment problems, such as student placement in New York City and Boston. However, we show that this fairness comes at the cost of extreme forms of inequality of allocation\footnote{The inequality discussed here is between matched pairs, and within each side of the market, not between the two sides of the market as in in \cite{GI:89} and \cite{AKL:2017}. In the latter, notions of equality and fairness relate to equalizing outcomes of both sides of the market while maintaining stability.}. In the spatial allocation case, this results in some students going to school across the street while other travel across the city.
The intuition is that students and schools that are close to each other can block any allocation that involves a pair that is further away, and peripheral or marginal students get the long end of the subway ride. 

We formalize this intuition by showing that the stable matching lexicographically maximizes the welfare of the matched pairs, starting with the best-off. We propose a simple algorithm that reflects this lexicographic ordering and makes the proof of our result transparent. We call this algorithm max-max-lex. Similarly, we propose an algorithm, adapted from the bottleneck algorithm in \cite{BDM:2009}, Section~6.2, that reverses the balance between stability and inequality, and matches pairs in lexicographic order maximizing the welfare of the worst-off. We call this algorithm max-min-lex. The resulting matching is Rawlesian at the expense of stability. 

\section{Model}
Consider a one-to-many matching problem with two sides
\emph{$\mathcal{I}$} and $\mathcal{J}$. We will call the
elements of $\mathcal{I}$ students, and the elements of
$\mathcal{J}$ schools. Let $\mathcal{J}$ be a discrete set with cardinality weakly smaller
than the cardinality of $\mathcal{I}$. Let each school $j \in \mathcal{J}$ have capacity~$q_{j}$, which is the number of students it is equipped to serve. Finally, let $u_{ij}$ be the utility of a student $i$ when matched with $j$, and similarly let $v_{jI}$ be the utility of a school $j$ when matched with a set of students $I\subseteq\mathcal{I}$. We normalize the utility of unmatched students to~$-\infty$. We assume that utilities are strictly positive, i.e.,~$u_{ij}>0$ for every $i$ and $j$; there are no indifferences, i.e.,
there are no pairs~$i,i'\in\mathcal{I}$ and~$j,j'\in \mathcal{J}$ such that~$u_{ij} = u_{i'j}$ or~$u_{ij} = u_{ij'}$,
and preferences are strictly aligned, by which we mean that for all~$j\in\mathcal J$ and~$I\subseteq\mathcal I$, 
$v_{jI} =  \sum_{i\in I} u_{ij}$.
Strictly aligned preferences are so called because they require alignment between the utilities of the two sides of the market. They are a particular type of altruistic preference. When the matching is one-to-one, the definition of strictly aligned preferences coincides with the definition of aligned preferences in \citet{NY:2009}.\footnote{The condition is related but stronger than the top coalition property in  \cite{BKS:2001} and weaker than the condition in \cite{Pycia:2012}.}

An allocation is a function $\mu:\mathcal{I}\cup \mathcal{J}\rightarrow 2^{\mathcal{I}}\cup \mathcal{J}$ such that $\mu(i)\in \mathcal{J}\cup \{i\}$ and $\mu(j)\subseteq \mathcal{I}\cup \{j\}$. The notation $\mu(i)=i$ indicates that student $i$ is unassigned, and $j \in \mu(j)$ indicates that the number of students assigned to school $j$ under $\mu$ is less than its capacity, that is $q_j>|\mu(j)\cap \mathcal{I}|$. 
An allocation is called feasible if each student is assigned to at most one school, and all school capacity constraints are respected, that is
if~$ |\mu(i)| = 1$ for all~$i \in \mathcal{I}$ and~$|\mu(j)| \leq q_{j}$ for all~$j \in \mathcal{J}.$
An allocation is stable when there are no blocking pairs. In our context, this is equivalent to the following. \begin{definition}
The allocation $\mu:\mathcal{I}\cup\mathcal{J}\rightarrow 2^{\mathcal{I}}\cup\mathcal{J} $
is stable if $\nexists\  i,j\in\mathcal{I}\times\mathcal{J}$ such that~$u_{ij} > u_{i\mu(i)}$
and~$\left[ \left[|\mu(j))|< q_{j}\right] \mbox{ or } \left[|\mu(j))|= q_{j} \mbox{ and } \exists i' \in \mu(j),\: \; u_{i'j} < u_{ij}\right]\right]$.
\end{definition}
The following algorithm will be shown to produce the unique stable matching. 
\begin{enumerate}
    \item Match Step: select $i$ and $j$ such that the utility of their match is the highest in the set of students that are unassigned and schools that have some residual capacity.
    \item Update Step: reduce the capacity of the school found in the previous step by 1. Remove the assigned student from the set of unassigned students.
\end{enumerate} 
We call this algorithm the max-max-lex algorithm\footnote{The max-max-lex algorithm is lexicographic, starting from the top. This feature is shared with rank-maximal allocations, see \citet{Irving2006}, where the number of agents receiving their first choice is maximized, subject to which a maximum number of remaining agents receive their second choice, etc...}
because it iteratively pairs the students and schools that are each other's top choice among the schools and students that are still available. It does so in a lexicographic order, until there are no further students and schools to match. The max-max-lex algorithm is formally described below. It converges in a finite number of steps. In the algorithm, we denote $e^j$ the $j$-the vector of the canonical basis of  $\mathbb{R}^{\mathcal{J}}$, which is the vector whose $j$-th entry is equal to one, and whose other entries are equal to zero. 

\begin{algorithm}%[h]
\setstretch{1.4}
\SetAlgoLined
 \vspace{2mm}
 Initialization:\\
 Set $t=0$, $I^0 = \mathcal{I}$ and $q^0 = q$\\
 \While{$I^t \neq \emptyset$ and $q^t \neq 0$}{
 $i^t,j^t = \argmax_{i,j} u_{ij}$\\
 s.t. $i\in I^t$ and $q^t_{j^t} \neq 0$ \\
 Set $\mu(i^t) = j^t$; \\
$q^{t+1} = q^t - e^{j^t}$;\\
 $I^{t+1} = I^{t} \setminus \{i^t\}$;\\
 $t=t+1$
 }
\caption{Max-max-lex Algorithm}
\end{algorithm}
\vspace{3mm}
Theorem \ref{thm} shows three important results: first, the allocation resulting from the max-max-lex algorithm is the one that maximizes the vector of students' utilities in lexicographic order from higher to lower utility pairs. Second, it proves that the allocation is stable. Finally, it shows that the stable allocation is unique, therefore implying that the resulting matching outcome of the max-max-lex algorithm is identical to the matching outcome of the DAA\footnote{\cite{CCP:2022} discuss the trade-off between (school) priorities and (student) preferences in school choice and show in particular that in the current context of aligned preferences, the stable outcome coincides with the top trading cycles algorithm of \cite{SS:74}. Hence, top trading cycles also produces high inequality in outcomes in this context.}.

\begin{thm}\label{thm} (a) The max-max-lex algorithm maximizes (among all feasible allocations) the vector of ranked ordered utilities of student-school pairs in the lexicographic order, starting from the pair with the highest utility.
(b) The assignment resulting from the max-max-lex algorithm is stable.
(c) The stable allocation is unique. 
\end{thm} \begin{proof}

(a) Let $\mathcal{U}\subseteq \mathbb{R}^{|\mathcal{I}|}$ represent the set of utilities that are achievable in the economy in a feasible allocation. Formally, let $u=(u_i)_{i\in \mathcal{I}}$ be a vector in $\mathbb{R}^{|\mathcal{I}|}$. If $u\in \mathcal{U}$ then there exists a feasible allocation $\mu$ such that $u_{i\mu(i)} = u_i$. Let $u^{(k)}$ represent the k-th order statistic of vector $u$, with $u^{(|\mathcal{I}|)}$ being the highest component of vector~$u$, and $u^{(1)}$ being its smallest. The first iteration of the max-max-lex algorithm selects among the vectors in $\mathcal{U}$ the ones with the highest value of $u^{(|\mathcal{I}|)}$. The n-th iteration of the max-max-lex algorithm selects among the vectors selected at the previous step, the ones with the highest value of $u^{(|\mathcal{I}|-n)}$, and so on. Therefore, the max-max-lex algorithm maximizes lexicographically the utility of students, starting from the pairs with the highest utility. 

(b) Let $\mu^{MML}$ be the match resulting from the max-max-lex algorithm, and assume by contradiction that it is unstable. This means that there exists $i$ and $j$ such that $u_{ij}>u_{i\mu(i)}$ and for some $i'\in \mu(j)$, $u_{ij}>u_{i'j}$. However, this implies that the max-max-lex algorithm would have matched $i$ and $j$, before matching $i'$ and $j$, which leads to a contradiction. 

(c) Let $\mu^{S}$ be a stable match and let $\mu^{MML}$ be the stable match arising from  the max-max-lex algorithm. Suppose by contradiction that $\mu^{S}\neq \mu^{MML}$. This means that there exists $i\in\mathcal{I}$ such that $\mu^{S}(i)\neq \mu^{MML}(i)$. Since by Assumption 2 there are no indifferences, it must be that either (a) $u_{i\mu^{MML}(i)}<u_{i\mu^{S}(i)}$ or (b) $u_{i\mu^{MML}(i)}>u_{i\mu^{S}(i)}$. First suppose that (a) holds. Since $i$ and $\mu^{S}(i)$ are not assigned through the max-max-lex algorithm, it must be that at the stage of the algorithm when $i$ is assigned, school $\mu^{S}(i)$ is already at full capacity. This implies that $\exists I\subseteq \mathcal{I} \mbox{ s.t. } |I|\geq q_{\mu^{S}(i)} \mbox{ and } \min_{i'\in I} u_{i'\mu^{S}(i)}>u_{i\mu^{S}(i)}$. But this implies that any $i'\in I$ would form a blocking pair with $\mu^{S}(i)$ in $\mu^{S}$. This contradicts that $\mu(s)$ is stable. Suppose then that (b) holds, i.e., $u_{i\mu^{MML}(i)}>u_{i\mu^{S}(i)}$. This implies that $\nexists I\subseteq\mathcal{I} \mbox{ s.t. } |I| \geq q_j \mbox{ and } \min_{i'\in I} u_{i'\mu^{MML}(i)}>u_{i\mu^{MML}(i)}$. But then $(i,\mu^{MML}(i))$ form a blocking pair in $\mu^{S}$, which is a contradiction. Therefore $\mu^{S} = \mu^{MML}$.
\end{proof}

An illustration of the severe inequality displayed by the stable allocation in matching with aligned preferences is given in Figure~2.1(a). The latter shows the stable matching between a large number of students uniformly distributed on~$[0,1]^2$ and~$5$ distinct schools in~$[0,1]^2$ with heterogeneous capacities. Utilities are spatial, i.e.,~$u_{ij}=\sqrt{2}-d_{ij}$, where~$d_{ij}$ denotes Euclidean distance between~$i$ and~$j$. 
%For illustrative purposes, Figure~2.1 actually represents the limit allocation when~$\mathcal I=[0,1]^2$. See \cite{HHP:2006} for details. 
Dots in the figure represent schools, and territories of the same color represent students who attend the same school. One characteristic of this assignment is that all schools lie in the territory that they serve. As one can see from the figure, some students in the red territory have to travel almost the maximum distance that can be traveled in the square, while others travel no distance at all. This results in very dispersed utilities in the stable allocation.

The lexicographic nature of the stable allocation suggests a Rawlesian alternative, where pairs are matched in lexicographic order, starting with the lowest utility pair within a set that is iteratively determined. The corresponding algorithm we propose below is adapted from the bottleneck algorithm in, for instance, \cite{BDM:2009}. The algorithm is made of three steps: 
\begin{enumerate}
    \item Feasibility Step: among the feasible allocations of unassigned students and schools, select one that maximizes the lowest utility $u^\star$ obtained by anyone in the allocation. 
    \item Match Step: match all $i$ and $j$'s in the feasible allocation identified at the previous step that obtain utility $u^\star$ from their match.
    \item Update Step: reduce the capacity of the school found in the previous step by 1. Remove the assigned student from the set of unassigned students.
\end{enumerate} 
This algorithm converges in finite time and produces an allocation that maximizes the utility of the worst-off student, by maximizing at each step the utility of the worst-off students among those remaining. We thus call this algorithm max-min-lex. 
% In the formal description of the algorithm below, we let $\mathcal{U} = (u_{ij})_{i\in\mathcal{I},j\in\mathcal{J}}$. Recall that the utility for an unassigned student is $-\infty$. 

\begin{algorithm}[h]
\setstretch{1.4}
\SetAlgoLined
 \vskip10pt
 Initialization:\\
 Set $t=0$, $I^0 = \mathcal{I}$, $q^0 = q$\\
 \While{$I^t \neq \emptyset$ and $q^t \neq 0$}{ $U^t=\{(u_{ij}): {i\in I^t, \;j:q_j^t\neq 0}\}$ \\
 By dichotomy, find the largest $u^\star$ in $U^t$ such that there exists a feasible match with no assigned student obtaining a utility below $u^{\star}$; \\
 Match $i^t,j^t$ such that $u_{i^tj^t}=u^\star$\\
 Set $I^{t+1} = I^{t} \setminus \{i^t\}$;\\
 $q^{t+1} = q^{t}-e^{j^t}$;\\
 $t=t+1$
}
 \vskip10pt
 \caption{Max-min-lex Algorithm}
\end{algorithm}
\vspace{3mm}

The equalitarian nature of max-min-lex allocations comes at the expense of stability. This is straightforward, given the uniqueness of the stable allocation. It also stems from the logic of the max-min-lex algorithm, which creates blocking pairs. It is most easily seen in a~$2$ students,~$2$ schools example, with~$u_{ij}>u_{ij'}>u_{i'j}>u_{i'j'}$. In this case, the max-max-lex algorithm matches~$(i,j)$ and~$(i'j')$, whereas the max-min-lex algorithm matches~$(i,j')$ and~$(i'j)$, thereby decreasing inequality but creating a blocking pair.

An illustration\footnote{See \cite{GV:2023} for details on the algorithm and implementation used to produce Figure~2.1.} of the max-min-lex allocation is given in Figure~2.1(b). The primitives are identical to those in Figure~2.1(a) but the allocation no longer displays the signs of extreme outcome inequality in Figure~2.1(a).

\begin{figure}%[h]
\label{Fair_allocation}
\begin{center}
\subfigure[Max-max-lex]{\includegraphics[scale=0.33]{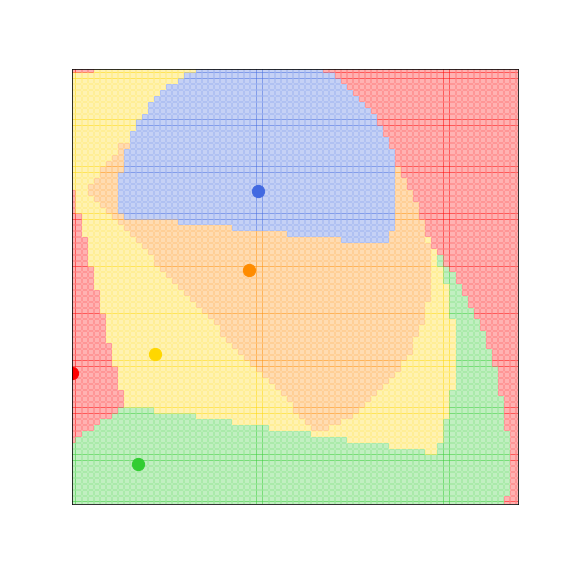}}
\subfigure[Max-min-lex]{\includegraphics[scale=0.33]{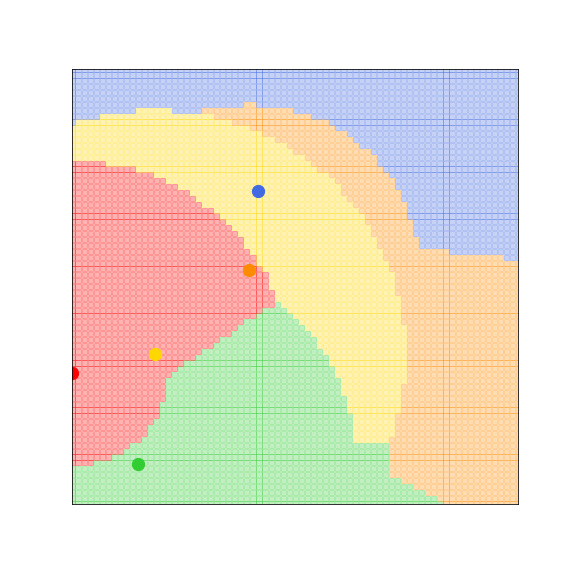}}
\end{center}
\caption{Allocations of students uniformly distributed on the unit square and~$5$ schools, represented by dots. Each dot serves the students in the territory of the corresponding color. The left panel shows the stable allocation, whereas the right panel shows the allocation resulting from the max-min-lex algorithm.}
\end{figure}

%\section{Conclusion}
%In conclusion, this paper shows that under certain conditions on preferences, the stable allocation leads to extreme inequality between the agents in the economy. Focusing on the case of aligned preferences, we make our result transparent by showing how the unique stable allocation is obtained by lexicographically matching pairs, starting from the best-off in the economy. We provide a visual representation of the inequality using the particular case of spatial preferences. Finally, we contrast with a solution that reverses the logic of the stable allocation, by lexicographically maximizing the welfare of the worst-off in the economy, at the expense of stability. 

%\nocite{*}
\bibliographystyle{chicago}
\bibliography{AntiRawles}

\end{document}